\documentclass[
 reprint, superscriptaddress,
 amsmath,amssymb,
 aps]{revtex4-2}
\pdfoutput=1 

\usepackage{graphicx}
\usepackage[dvipsnames]{xcolor}
\usepackage{subcaption}
\usepackage{dcolumn}
\usepackage{bm}
\usepackage{bbold}
\usepackage{physics}
\usepackage{amsmath,amsthm,amssymb,amsfonts,physics,bbm,graphicx}
\usepackage{xfrac, hyperref}
\hypersetup{colorlinks,citecolor=blue,filecolor=blue,linkcolor=blue,urlcolor=blue}
\usepackage{natbib}
\usepackage{mathtools}
\usepackage{comment}
\usepackage{enumerate}

\newtheorem{theorem}{Theorem}
\newtheorem{lemma}{Lemma}
\newtheorem{defn}{Definition}

\newcommand{\ba}{\begin{eqnarray}}
\newcommand{\ea}{\end{eqnarray}}
\newcommand{\ban}{\begin{eqnarray*}}
\newcommand{\ean}{\end{eqnarray*}}

\newcommand{\pg}{P_{\textnormal{guess}}}
\newcommand{\pgs}{P_{\textnormal{guess}}^{*}}
\newcommand{\id}{\mathbb{1}}

\definecolor{darkgreen}{rgb}{0,0.5,0}

\newcommand{\hmax}{{H_{\textnormal{max}}}}
\newcommand{\lmax}{{\lambda_{\textnormal{max}}}}
\newcommand{\psecr}{{p_{\textnormal{secr}}}}
\newcommand{\vmax}{{v_{\textnormal{max}}}}
\newcommand{\umax}{{u_{\textnormal{max}}}}

\usepackage[normalem]{ulem}

\allowdisplaybreaks

\begin{document}

\title{Maximal intrinsic randomness of a quantum state}

\author{Shuyang Meng}
\affiliation{Department of Physics, National University of Singapore, 2 Science Drive 3, Singapore 117542} 

\author{Fionnuala Curran}
\affiliation{ICFO-Institut de Ci\`encies Fot\`oniques, The Barcelona Institute of Science and Technology,
Av. Carl Friedrich Gauss 3, 08860 Castelldefels (Barcelona), Spain}

\author{Gabriel Senno}
\affiliation{Quside Technologies S.L., C/Esteve Terradas 1, 08860 Castelldefels,
Barcelona, Spain}

\author{Victoria J. Wright}
\affiliation{ICFO-Institut de Ci\`encies Fot\`oniques, The Barcelona Institute of Science and Technology,
Av. Carl Friedrich Gauss 3, 08860 Castelldefels (Barcelona), Spain}

\author{M\'at\'e Farkas}
\affiliation{Department of Mathematics, University of York, Heslington, York, YO10 5DD, United Kingdom}
\affiliation{ICFO-Institut de Ci\`encies Fot\`oniques, The Barcelona Institute of Science and Technology,
Av. Carl Friedrich Gauss 3, 08860 Castelldefels (Barcelona), Spain}

\author{Valerio Scarani}
\affiliation{Centre for Quantum Technologies, National University of Singapore, 3 Science Drive 2, Singapore 117543}
\affiliation{Department of Physics, National University of Singapore, 2 Science Drive 3, Singapore 117542} 

\author{Antonio Ac\'{\i}n}
\affiliation{ICFO-Institut de Ci\`encies Fot\`oniques, The Barcelona Institute of Science and Technology,
Av. Carl Friedrich Gauss 3, 08860 Castelldefels (Barcelona), Spain}
\affiliation{ICREA - Instituci\'o Catalana de Recerca i Estudis Avan\c{c}ats, 08010 Barcelona, Spain}

\date{\today}

\begin{abstract}
One of the most counterintuitive aspects of quantum theory is its claim that there is `intrinsic' randomness in the physical world. Quantum information science has greatly progressed in the study of intrinsic, or secret, quantum randomness in the past decade. With much emphasis on device-independent and semi-device-independent bounds, one of the most basic questions has escaped attention: how much intrinsic randomness can be extracted from a given state $\rho$, and what measurements achieve this bound? We answer this question for three different randomness quantifiers: the conditional min-entropy, the conditional von Neumann entropy and the conditional max-entropy. For the first, we solve the min-max problem of finding the projective measurement that minimises the maximal guessing probability of an eavesdropper. The result is that one can guarantee an amount of conditional min-entropy $H^*_{\textrm{min}}=-\log_2\pgs(\rho)$ with $\pgs(\rho)=\frac{1}{d}\,\left(\tr \sqrt{\rho}\right)^2$ by performing suitable projective measurements.  For the conditional von Neumann entropy, we find that the maximal value is $H^{*}= \log_{2}d-S(\rho)$, with $S(\rho)$ the von Neumann entropy of $\rho$, while for the conditional max-entropy, we find the maximal value $H^{*}_\textnormal{max}=\log_{2}d + \log_{2}\lmax(\rho)$, where $\lmax(\rho)$ is the largest eigenvalue of $\rho$. Optimal values for $H^{*}_{\textrm{min}}$, $H^{*}$ and $H^{*}_\textnormal{max}$ are achieved by measuring in any basis that is unbiased with respect to the eigenbasis of $\rho$, as well as by other, less intuitive, measurements.
\end{abstract}

\maketitle

\section{Introduction}

One of the core differences between classical and quantum physics is the latter's probabilistic character, which is irreducible to ignorance of underlying variables. This difference has fundamental implications for our worldview, but it is also attractive as a natural source of randomness for practical uses. Indeed, Geiger counting was already used as a source of physical randomness in the second half of the 20th century. In the past two decades, with the development of quantum information science, a large number of quantum random number generators (QRNGs) have been designed, and many have been implemented, usually with light (see \cite{RevModPhys.89.015004,review2022} for comprehensive reviews). The \textit{amount of randomness} is naturally captured by the \textit{guessing probability} $\pg$: the higher the probability that the random variable is guessed, the smaller the randomness. This intuitive characterisation was found to have operational meaning: the \textit{min-entropy} $H_{\textrm{min}}=-\log_2\pg$ quantifies (informally) the fraction of perfect coin tosses that can be extracted from a string generated by the available source. But randomness is not an absolute notion: one has to specify \textit{for whom} the source should be partly unpredictable. For mere sampling purposes, it might be sufficient to take the observed probabilities at face value; for cryptographic applications, however, one needs to estimate the probability that \textit{an adversary, Eve,} guesses the outcomes. The resulting randomness is called \textit{secret randomness}, or \textit{intrinsic randomness}.

The computation of intrinsic randomness using quantum resources and against a quantum adversary has been studied from different perspectives. When considering a user with classical data correlated with quantum information in the hands of an adversary, the min-entropy quantifies the amount of perfect random bits that the user can establish~\cite{Konig_2009}. The question was also addressed for the task of quantum key distribution, which is the extraction of secret \textit{shared} randomness. It was in this context that the idea of device-independent certification was born: the possibility of bounding the amount of randomness in a black-box setting, based on the observation of Bell-nonlocal correlations \cite{DI2007}. Next, it was noticed that device-independent certification can be performed for randomness as well \cite{colbeckPhD,Nature2010}, providing the first disruptive case for quantum randomness in a non-shared setting \footnote{As long as \emph{process randomness} requires characterised devices, classical and quantum RNGs compete on the same grounds for speed, stability, practicality etc. But, given an alleged RNG as a black box, on classical devices, one can only test product randomness with statistical tests. While process randomness implies product randomness, the opposite is certainly not true: one could have recorded a long enough list of random numbers, and the device under study may just be reading deterministically from that record.}. This breakthrough happened as the race to demonstrate loophole-free Bell tests was taking up speed. There followed an explosion of designs and implementations of QRNGs certifiable under various assumptions, from device-independent (disruptive, but hard to implement), to semi-device-independent in various forms, to fully characterised (practical and fast, but requiring a precise modelling of the setups). For these developments, we refer to the reviews \cite{RevModPhys.89.015004,Bera_2017,reviewAM,review2022}.

In this flurry of activity, one of the most basic questions
was somehow left out: \textit{how much secret randomness can be extracted from a known state $\rho$}. In this paper, we solve this problem for three of the most natural and operational measures of randomness: the conditional min-entropy, the conditional von Neumann entropy and the conditional max-entropy. For the first, we show that the answer is $H^*_{\textrm{min}}=-\log_2\pgs(\rho)$, with
\ba\label{eqn: main}
   \pgs(\rho)=\frac{1}{d}\,\big(\tr \sqrt{\rho}\big)^2\,,
\ea 
where $d$ is the dimension of the Hilbert space of the system, assumed to be finite. We find a family of measurements that generate this amount of randomness, which is closely related to the concept of `pretty good measurements'~\cite{hausladen1994pretty}, originally used as a close-to-optimal way to distinguish an ensemble of states. For the second, we find the maximal value
\begin{align}\label{eqn: H_star}
    H^{*}= \log_{2}d -S(\rho)\,,
\end{align}
where $S(\rho)=-\tr\Big(\rho \log_2 \rho \Big)$ is the von Neumman entropy of $\rho$, while for the third, we find
\begin{align}\label{eqn: H_max_star}
    H^{*}_{\textnormal{max}}= \log_{2}d +\log_2\lmax(\rho)\,,
\end{align}
where $\lmax(\rho)$ is the largest eigenvalue of $\rho$.
Interestingly, for $d>2$, we find that some measurements maximise one of $H_{\textrm{min}}$, $H$ and $H_{\textnormal{max}}$, but not the other two.

\section{Qubit example}
A case study will help to introduce the main ideas. Alice has a source that produces a qubit. She has characterised its state to the best of her knowledge and found it to be
\begin{equation}
\rho=\frac{1}{2}(\id+m\sigma_z)=\frac{1+m}{2}\ketbra{0}+\frac{1-m}{2}\ketbra{1}    
\end{equation}
for some $0\leq m\leq 1$. If she measures $\sigma_x$, her observed statistics will be those of a perfect unbiased coin: $P_A(+1)=P_A(-1)=\frac{1}{2}$. Suppose now that what the source really does is produce a pure state in each round, specifically half of the rounds $\ket{\chi_+}$ and half of the rounds $\ket{\chi_-}$, with 
\begin{equation}
    \ket{\chi_{\pm}}=\sqrt{\frac{1\pm\sqrt{1-m^2}}{2}}\ket{+x}+\sqrt{\frac{1\mp\sqrt{1-m^2}}{2}}\ket{-x}
\end{equation} (indeed, $\frac{1}{2}\ket{\chi_+}\bra{\chi_+}+\frac{1}{2}\ket{\chi_-}\bra{\chi_-}=\rho$). If Eve knows the working of the source exactly, she will guess $i=+1$ ($i=-1$) in the rounds when the source sent out $\ket{\chi_+}$ ($\ket{\chi_-}$). Her guess will then be correct with probability
\begin{equation}\label{eqn: optex}
    \pg=\frac{1}{2}\Big(1+\sqrt{1-m^2}\Big)\,,
\end{equation} 
which is strictly larger than $\frac{1}{2}$ when $m<1$ (i.e.~when $\rho$ is mixed). Thus, the intrinsic randomness of Alice's protocol is less than her apparent perfect randomness. In particular, there is no secret randomness in the state $\rho=\frac{1}{2}\id$, since $\pg=1$ for $m=0$.

As will be expanded on in what follows, two things are already known about this case study and its generalisation to higher dimensions. First: we presented this example with Eve having perfect classical information about the source, in the sense that she knows at each instance which state has been prepared and accordingly makes her guess on Alice's measurement outcome. However, the result is unchanged if Eve holds quantum side-information. Eve then holds a purification of Alice's state, and she measures her own system to guess Alice's result. Since the two scenarios are equivalent in terms of the guessing probability, we will move from one to the other when convenient for the argumentation. Second: having fixed Alice's protocol (both the state and the measurement), the maximisation of $\pg$ over all decompositions of $\rho$ is a known semidefinite program (SDP) \cite{Law_2014}; in the case study, we have presented the optimal decomposition. What is not known is whether $\sigma_x$ is the best measurement for Alice, even in the presence of Eve: could another measurement on the same state $\rho$ decrease Eve's guessing probability, at the expense of biasing the observed $P_A$? 
We set out to solve this min-max problem, and thus determine the maximal amount of secret randomness that can be extracted from $\rho$.

\section{Setting of the problem}
Alice holds a quantum state $\rho$ from a Hilbert space of dimension $d$. We want to determine how much intrinsic randomness she can extract from $\rho$ and which measurement achieves this maximum. We consider only
measurements $\mathcal{M}=\{M_i\}_i$ which are projective, i.e. $M_iM_j= \delta_{ij}M_i$, where $\delta_{ij}$ is the Kronecker delta (we discuss general POVMs at the end of this section). To quantify how intrinsically random, that is, how unpredictable, Alice's measurement outcome is, one considers the existence of an eavesdropper, Eve, who has a more detailed knowledge than Alice about the process, but cannot actively influence it (she is `outside the lab'). Concretely, in every round, Eve knows the true state $\rho_c$ produced by the source. Given this knowledge, she guesses the most likely outcome $i=i(c)$ for that round. Without loss of generality, we can group together all of Eve's states that lead to the same guessed outcome, since Eve does not gain anything in treating them as distinct. We denote by $\rho_i$ the states seen by Eve, sub-normalised such that $q_i=\tr\rho_i$ is the probability that Eve's most likely outcome is $i$. These states must satisfy $\sum_i \rho_i=\rho$.

Having set this stage, Eve's average guessing probability is $P_{\textrm{guess}}\big(\{\rho_i\},\mathcal{M}\big)=\sum_{i}\tr(M_{i}\rho_i)$. Since we don't know the true states $\rho_i$, we need to consider the worst case scenario, i.e.~the decomposition that maximises Eve's guessing probability,
\ba\label{eqn: sdp}
P_{\textrm{guess}}(\rho,\mathcal{M})&=&\max_{\{\rho_i\}}\sum_i \tr(M_i\rho_i)\\&\textrm{s.t.}& \rho_i\geq 0\,,\;\sum_i \rho_i=\rho\,. \nonumber\ea
This optimisation is an SDP, and so can be solved efficiently. In order to determine, \textit{the maximal amount of secret randomness that can be extracted from the known state $\rho$}, one need to optimize Eq. \eqref{eqn: sdp} over Alice's measurement, i.e.~compute
\begin{equation}\label{eqn: defpg} 
P_{\textrm{guess}}^*(\rho)=\min_{\mathcal{M} \in \Pi}P_{\textrm{guess}}(\rho,\mathcal{M})\,,
\end{equation}
where $\Pi$ is the set of all projective measurements.
Our main result is to show that Eq.~\eqref{eqn: main} is the solution to the optimisation \eqref{eqn: defpg}.

The search for an optimal measurement could have been extended to the larger set of Positive Operator-Valued Measures (POVMs), but the operational interpretation in our context is unclear. Recall that our goal is to quantify the secret randomness \textit{in the state $\rho$}. When implementing a POVM, however, the projective measurement acts on the given state $\rho$ plus an auxiliary system, so part of the obtained randomness may come from the latter. In fact, for extremal measurements minimising the guessing probability, the auxiliary system has to be in a pure state, say $\ket a$, of dimension $d_A$ \cite{Senno_2023}. It follows from our main result that the maximal amount of randomness obtained when implementing a projective measurement on the global state is $\pgs\big(\rho\otimes\ketbra a\big)$. It is easy, however, to see that $\pgs\big(\rho\otimes\ketbra a\big)=\frac{1}{d_A}\pgs(\rho)$, that is, the optimal guessing probability is equal to that obtained by performing the corresponding optimal projective measurements independently on the system and the auxiliary.

Thus, the extra randomness supplied by using 
the optimal POVM is exactly equal to the 
intrinsic randomness of the auxiliary system used to implement the POVM, so we view
it as arising from the auxiliary rather than from $\rho$ itself.

\section{Main results}
\begin{theorem}\label{thm: main}
The maximal amount of secret randomness that can be extracted from a quantum state $\rho$ using a projective measurement is given by $H^*_{\textrm{min}}=-\log_2 \pgs(\rho)$ with $\pgs(\rho)=\frac{1}{d}\,\left(\tr \sqrt{\rho}\right)^2$.
\end{theorem}
Without loss of generality, one can restrict the optimisation to rank-one projective measurements (see Appendix \ref{app:rank1}). In what follows, we outline a proof that uses notions from state discrimination and the resource theory of coherence, with full details in Appendix \ref{app:proof1}. An alternative proof using properties of the min-entropy and semidefinite programming is provided in Appendix \ref{app:proof2}. We prove the theorem by first proving the lower bound $\pgs(\rho)\geq\frac{1}{d}\,\left(\tr \sqrt{\rho}\right)^2$ (Lemma \ref{mainlemma1}) and then showing that there exist measurements that achieve that bound (Lemma \ref{lemma: measure}). 
\begin{lemma}\label{mainlemma1}
    The lower bound $\pgs(\rho)\geq\frac{1}{d}\,\left(\tr \sqrt{\rho}\right)^2$ holds for every state $\rho$.
\end{lemma}
\begin{proof} 
Using the fact that rank-one measurements are optimal for Alice, from \cite[Theorem 1, (iii)]{Coles_2012}, we find 
\begin{align}
 P_{\textrm{guess}}(\rho,\mathcal{M})=&\max_{\{\sigma\in\mathcal{I}_{\mathcal{M}}\}} F(\rho,\sigma)\,,\label{PG2}  
\end{align}
where $F$ is the Uhlmann fidelity and $\mathcal{I}_{\mathcal{M}}$ is the set of states that are diagonal in the measurement basis $\{\ket{m_i}\}$.  Notice that $\mathbb{1}/d\in \mathcal{I}_{\mathcal{M}}$ for all $\mathcal{M} \in \Pi$, so $P_{\textrm{guess}}(\rho,\mathcal{M})\geq F(\rho,\mathbb{1}/d)=\frac 1 d (\tr\sqrt{\rho})^2$ for all $\mathcal{M}$. Hence, $\pgs(\rho)$ cannot be smaller than $\frac 1 d (\tr\sqrt{\rho})^2$. 
\end{proof}

\begin{lemma}\label{lemma: measure}
A projective measurement $\mathcal{M}$ in the basis $\{\ket{m_i}\}$ achieves the bound $P_{\textnormal{guess}}(\rho,\mathcal{M})=\frac 1 d (\tr\sqrt{\rho})^2$ if and only if $\matrixel{m_i}{\sqrt{\rho}}{m_i}=\frac{1}{d}\tr\sqrt{\rho}$ for all $i=1,...,d$.
\end{lemma}

\begin{proof}
The details missing here are provided in Appendix \ref{app:proof11}. In the quantum side-information scenario, any $k$-outcome rank-one measurement, $\mathcal{M}$, by Alice steers $k$ pure states on Eve. To optimise her guess, Eve has to measure her system to optimally discriminate among these $k$ states. It is known \cite{eldar2003designing} that the best discrimination of a set of pure states is obtained with rank-one measurements: thus, Eve will also perform a rank-one measurement. In turn, Eve's measurement defines an ensemble realising the mixed state $\rho$. This ensemble
consists of $k$ pure states $\rho_i=\dyad{\tilde{\psi_i}}$. If $\mathcal{M}$ is projective, we have $k=d$ and it follows from \cite{hughston1993complete} that any decomposition of $\rho$ in $d$ pure states is defined by the choice of an orthonormal basis $\{\ket{i}\}$ through \ba\label{decompEve}
\tilde{\ket{\psi_i}}=\sqrt{\rho}\ket{i}&\textrm{ with }& \braket{i}{i'}=\delta_{ii'}\,,\;i,i'=1,...,d\,.
\ea
Inserting all these observations in \eqref{eqn: sdp}, 
we obtain
\begin{align}\label{PG1}
P_{\textrm{guess}}(\rho,\mathcal{M})=\max_{\{\ket{i}\}}  \sum_i \abs{\matrixel{m_i}{\sqrt{\rho}}{i}}^2 \,.
\end{align}
The r.h.s.~has been called the geometric coherence of $\rho$ \cite{streltsov2015measuring} and was shown in \cite{xiong2018geometric} to be equivalent to $\max_{\{\sigma\in\mathcal{I}_{\mathcal{M}}\}} F(\rho,\sigma)$, with $\mathcal{I}_{\mathcal{M}}$ the set of states diagonal in the basis $\{\ket{m_i}\}$.
If we rewrite \eqref{PG1} as
\begin{align}
P_{\textrm{guess}}(\rho,\mathcal{M})&=&\max_{\{\Pi_i\}_i}\sum_i \tr\big(\Pi_i\dyad{\tilde{\gamma_i}}\big) \label{PG5}\\&\textrm{s.t.}& \Pi_i\geq 0\,,\;\sum_i \Pi_i=\id\,,\nonumber
\end{align} 
the r.h.s~defines the optimal discrimination of the sub-normalised states $\ket{\tilde{\gamma_i}}:=\sqrt{\rho}\ket{m_i}$ with a projective measurement $\Pi_i=\dyad{i}$. One then checks (see Appendix \ref{app:thm2}) that, under the assumption that
\ba\label{eqn: equality}
\matrixel{m_i}{\sqrt{\rho}}{m_i}=\frac{1}{d}\tr\sqrt{\rho}&\textrm{for all}& i=1,...,d\,,
\ea
the choice $\ket{i}=\ket{m_i}$ fulfills all the conditions for optimal discrimination of the $\ket{\tilde{\gamma_i}}$ \cite{holevo1973statistical,helstrom1969quantum,yuen1975optimum}. Thus, for measurements satisfying \eqref{eqn: equality}, it holds that $P_{\textrm{guess}}(\rho,\mathcal{M})=\sum_i \abs{\matrixel{m_i}{\sqrt{\rho}}{m_i}}^2=\frac{1}{d}(\tr\sqrt{\rho})^2$. Furthermore, we prove (see Appendix \ref{app: nec_cond}) that the condition \eqref{eqn: equality} is also necessary for a projective measurement to achieve the optimal guessing probability. What remains to be proven is that there exist measurements satisfying condition \eqref{eqn: equality}. An example of such a measurement valid for any state is the one defined by a basis $\{\ket{m_i}\}$ that is unbiased to the eigenbasis of $\rho$, that is, all moduli of inner products between elements of the two different bases equal $\frac{1}{\sqrt{d}}$.
However, as we discuss in Section \ref{sec: case_studies}, one can find other measurements satisfying condition \eqref{eqn: equality}
when $d>2$.
\end{proof}

Notice that, when Alice uses measurements satisfying \eqref{eqn: equality}, the decomposition \eqref{decompEve} that is optimal for Eve is $\tilde{\ket{\psi}_i}=\sqrt{\rho}\ket{m_i}$. If $\rho$ is full rank, $M_i=\rho^{-1/2}\rho_i\rho^{-1/2}$ is the `pretty good measurement' \cite{hausladen1994pretty} for the ensemble $\{q_i, \rho_i/q_i\}$ steered by Eve. This measurement is known to be optimal when special symmetries like \eqref{eqn: equality} are present in the problem \cite{dalla2015optimality} (in the notation of that work, the Gram matrix has entries $G_{ij}=\matrixel{m_i}{\rho}{m_j}$). Moreover, when Alice's measurement satisfies \eqref{eqn: equality} and when $\rho$ is full-rank, we can show (see Appendix \ref{app: Eve_discrim}) that Eve's optimal measurement to discriminate her local states is also a `pretty good' measurement.

After solving the problem for the guessing probability, we now move to the von Neumann entropy of the measurement outcomes conditioned on Eve's side information, a quantity of relevance in the multi-round setting~\cite{tomamichel2009fully,dupuis2020entropy,dai2023intrinsic}.

\begin{theorem}\label{thm: H}
The maximal conditional entropy that can be extracted from a quantum state $\rho$ using a projective measurement is $H^*=\log_2 d-S(\rho)$, where $S(\rho)=-\tr \rho \log_2{\rho}$ is the von Neumann entropy.
\end{theorem}

\begin{proof}
From \cite[Theorem 1, (i)]{Coles_2012}, we have that the entropy $H(Z|E)$ of Alice's measurement outcomes $Z$ conditioned on Eve's side information $E$ is
\begin{align}
    H(Z|E)= D\big(\rho\,||\sum_{z}M_z \rho M_z\big)\,,
\end{align}
where $\{M_z\}_z$ is Alice's projective measurement and $D(\rho||\sigma)$ is the quantum relative entropy between the states $\rho$ and $\sigma$,
\begin{align}
    D(\rho||\sigma)= \tr \Big( \rho \big(\log_2{\rho}-\log_2{\sigma} \big)\Big)\,,
\end{align}
which is defined when the support of $\rho$ is contained within the support of $\sigma$. In Appendix \ref{app:rank1_H}, we show that: 1) a rank-one measurement is optimal for Alice to maximise $H(Z|E)$ for a given $\rho$, and 2) that 
\begin{align}\label{eqn: S_ents}
 D\big(\rho\,||\sum_{z}M_z \rho M_z\big)= S\big(\sum_{z}M_z \rho M_z\big)-S(\rho)\,.   
\end{align}
In \cite{yuan2019quantum}, the r.h.s. is shown to be equivalent to the relative entropy of coherence of $\rho$ with respect to the measurement basis, which is used as a quantifier of randomness. 

\medskip

The maximum von Neumann entropy of a state of dimension $d$ is $\log_2 d$ and is achieved only for maximally mixed states, so we can upper bound Eq. \eqref{eqn: S_ents} with 
\begin{align}
\label{eqn: vnbound}
  H(Z|E) \leq \log_2 d -S(\rho)\,,  
\end{align}
with equality reached if and only if Alice's measurement basis $\{\ket{m_z}\}_z$ leaves her system in the maximally mixed state, i.e. if the condition 
\ba\label{eqn: H equality}
\matrixel{m_z}{{\rho}}{m_z}=\frac{1}{d}&\textrm{for all}& z=1,...,d\,
\ea
is satisfied. 
\end{proof}
As in the case for the condition \eqref{eqn: equality} for $H^{*}_\textrm{min}$, suitable measurements satisfying \eqref{eqn: H equality} include bases $\{\ket{m_z}\}$ that are unbiased to the eigenbasis of $\rho$, implying the tightness of \eqref{eqn: vnbound}. However, when $d >2 $ we can find other suitable measurements, as discussed in Section \ref{sec: case_studies}. The quantity $\log_2 d -S(\rho)$ is defined in \cite{zurek2001information} as the total information of $\rho$ and it is used in \cite{Horodecki_2003} as a measure of the objective information of $\rho$.

We now consider the conditional max-entropy of the measurement outcomes conditioned on Eve's side information. This quantity has been interpreted as the security of Alice's measurement outcomes when used as a secret key \cite{Konig_2009}.

\begin{theorem}\label{thm: Hmax}
The maximal conditional max-entropy that can be extracted from a quantum state $\rho$ using a projective measurement is $H^*_{\textnormal{max}}=\log_2 d +\log_2\lmax(\rho)$, where $\lmax(\rho)$ is the largest eigenvalue of $\rho$.
\end{theorem}

\begin{proof}
The details missing here are given in Appendix \ref{app: Hmax_details}. Without loss of generality, we restrict Alice to performing rank-one projective measurements (see Appendix \ref{app:rank1}). In the case where Alice makes a rank-one projective measurement, the conditional max-entropy of her outcomes conditioned on Eve can be formulated \cite{Konig_2009} as
\begin{align}
    \hmax\left (A|E \right)= \log_2 \psecr\,,
\end{align}
where
\begin{align}\label{eqn: psecr_trace}
\psecr
&= \max_{\sigma} \left( \sum_{x}  \sqrt{{p_x}\tr\left( { \sigma \ketbra{\psi_{x}^{E}}{\psi_{x}^{E}}}\right) } \right)^2
\\
&\textrm{s.t.} \quad \sigma \geq 0\,,\;\tr \sigma =1 \,,  \end{align}
where $\{\ket{\psi_{x}^{E}}\}$ are Eve's post-measurement states and $p_x= \langle m_x|\rho |m_x \rangle$. By applying the Cauchy-Schwartz inequality and identifying the semidefinite optimisation problem for the maximum eigenvalue of a quantum state, we find
\begin{align}\label{eqn: max_psec}
 \psecr \leq d \lmax(\rho)\,.   
\end{align}
In the case where the largest eigenvalue of $\rho$ is unique, the bound \eqref{eqn: max_psec} is reached if and only if the condition 
\ba\label{eqn: Hmax equality}
{\lvert \langle m_x  | \umax\rangle   \rvert}^{2}=\frac{1}{d}&\textrm{for all}& x=1,...,d\,
\ea
is satisfied, where $\ket{\umax}$ is the eigenvector of $\rho$ corresponding to its largest eigenvalue. The optimal measurements in the case where the maximum eigenvalue of $\rho$ is degenerate are discussed in Appendix \ref{app: Hmax_details}.  
\end{proof}
As in the case of $H^{*}_\textrm{min}$ and $H^{*}$, suitable measurements satisfying \eqref{eqn: Hmax equality} include bases $\{\ket{m_x}\}$ that are unbiased to the eigenbasis of $\rho$, but, as before, when $d >2 $ we can find other suitable measurements, as discussed in Section \ref{sec: case_studies}.

\section{Two case studies}\label{sec: case_studies}
Let us now study measurements that satisfy \eqref{eqn: equality}, \eqref{eqn: H equality} or \eqref{eqn: Hmax equality} but which are not unbiased to the eigenbasis of $\rho$. For one qubit, it is quickly verified that all measurements that satisfy \eqref{eqn: equality} are unbiased, so our first case study is for \textit{one qutrit}. Consider $\rho=\sum_{i=1}^3\lambda_i \dyad{i}$ with $\lambda_1 \geq \lambda_2 \geq \lambda_3$, and the measurement basis $\{M_i=\dyad{m_i}\}_{i=1,2,3}$, with 
\begin{align}\label{eqn: qutrit_meas}
    \ket{m_1 }&=\sqrt{\frac{1+a}{3}}\ket{1}+\sqrt{\frac{1+b}{3}}\ket{2}+\sqrt{\frac{1+c}{3}}\ket{3}\,, \nonumber
    \\
    \ket{m_2 }&=\sqrt{\frac{1+a}{3}}\ e^{i\theta_1}\ket{1}+\sqrt{\frac{1+b}{3}}\ket{2}+\sqrt{\frac{1+c}{3}}\ e^{i\theta_2}\ket{3}\,,
\end{align}
and $\ket{m_3}$ defined by the normalization condition $\sum_i M_i=\id$, 
where $a = -({\gamma_2}-{\gamma_3})k,\, b = ({\gamma_1}-{\gamma_3})k,\, c = -({\gamma_1}-{\gamma_2})k,\,  k\in \mathbb{R}$ and each $\gamma_i \geq 0$ with $\gamma_1 \geq \gamma_2 \geq \gamma_3$. We show in Appendix \ref{app: qutrit} that suitable parameters $\theta_1$ and $\theta_2$ can always be chosen such that this is a valid rank-one projective measurement when $k$ is in the range $-\frac{1}{2}\leq k\leq \frac{1}{2}$. 

This measurement basis is not in general unbiased to the eigenbasis of $\rho$, except when $\rho$ is maximally mixed. When we set $\{\gamma_i\}=\{\sqrt{\lambda_i}\}$, it is straightforward to show that the condition  \eqref{eqn: H equality} for the measurement to maximise $H_{\textnormal{min}}$ is satisfied. Similarly, if we set $\{\gamma_i\}=\{\lambda_i\}$, we see that the condition \eqref{eqn: H equality} for maximal $H$ is satisfied. Finally, the condition \eqref{eqn: Hmax equality} for maximal $\hmax$ is satisfied when $\gamma_2=\gamma_3$, so we see that, for qutrits at least, there exist non-unbiased measurements that achieve maximal randomness for every $\rho$ for all three of our quantifiers of randomness. Interestingly, though, these three conditions are inequivalent in general, so one can choose parameters $\{\gamma_i\}$ such that the measurement maximises any one of the entropies but not the other two.

The second case study uses \textit{two qubits}. It is based on the observation (proved in Appendix \ref{app:prodbasis}) that there is no product basis unbiased to the basis 
\begin{equation}
    \begin{aligned}\label{eq:entbasistext}
    \ket{\psi_1}&=\ket{00}\\
    \ket{\psi_2}&=\frac{1}{\sqrt{3}}(\ket{01}+\ket{10}+\ket{11})\\
    \ket{\psi_3}&=\frac{1}{\sqrt{3}}(\ket{01}+\omega\ket{10}+\omega^2\ket{11})\\
    \ket{\psi_4}&=\frac{1}{\sqrt{3}}(\ket{01}+\omega^2\ket{10}+\omega\ket{11})\,,
    \end{aligned}
\end{equation} where $\omega=e^{i\,2\pi/3}$. Consider a state $\rho=\sum_{k=1}^4 \lambda_k \dyad{\psi_k}$ diagonal in this basis. To extract the maximal randomness with an unbiased measurement, one must be able to perform entangled measurements. This is not a conceptual problem in our setting, since there is no reason why the two qubits should be far apart; nonetheless, such measurements may be more challenging to perform than basic single-qubit measurements. The question is: can one extract maximal randomness from $\rho$ by using a product basis? The answer seems to be positive. While we do not have an analytical proof, for a large number of choices of $\lambda$, we performed a heuristic optimisation over product bases, both general ($\{\ket{a,b},\ket{a,b^\perp},\ket{a^\perp,c},\ket{a^\perp,c^\perp}\}$, with six free parameters) and restricted to proper product measurements ($\{\ket{a,b},\ket{a,b^\perp},\ket{a^\perp,b},\ket{a^\perp,b^\perp}\}$, with four free parameters). In both cases and for all states that we probed, we numerically found measurements satisfying $\sum_{i} \Big( \langle m_i | \sqrt{\rho}|m_i \rangle - \frac{\tr \sqrt{\rho}}{4} \Big)^{2} \leq 10^{-15}$, $\sum_{i} \Big( \langle m_i | {\rho}|m_i \rangle - \frac{1}{4} \Big)^{2} \leq 10^{-15}$ or $\sum_{i} \Big( \abs{\langle m_i | \umax \rangle}^2 - \frac{1}{4} \Big)^{2} \leq 10^{-15}$, which suggests that there exist product measurements satisfying the conditions \eqref{eqn: equality}, \eqref{eqn: H equality} and \eqref{eqn: Hmax equality}, respectively.
In this family of examples, therefore, the freedom to choose a measurement basis that is not unbiased may lead to a practical advantage: it allows one to obtain maximal randomness with product measurements.

\section{Conclusion}
It is well known that quantum physics contains an intrinsic form of randomness, but, somewhat surprisingly, given a quantum state, it is unknown what is the optimal measurement to extract from it the maximum amount of such randomness. In this work, we concentrate on three different quantifiers of the amount of randomness in a measurement's outcomes conditioned on an adversary's side information: the conditional min-entropy, the conditonal von Neumann entropy and the conditional max-entropy.
As one might have expected, all measurements in a basis that is unbiased to the eigenbasis of $\rho$ maximise all three of these conditional entropies. However, we also find other measurements that achieve the optimal values, providing a flexibility that may have practical implications, as in the second case study reported. In fact, beyond its fundamental motivation, our analysis is also relevant for the design of device-dependent QRNGs, for which the quantum state is fully characterised. Interestingly, we find measurements in the qutrit case that maximise one of the three conditional entropies considered, but which are not optimal for the other two.

\section*{Acknowledgments}
We thank Siddhartha Das for pointing out to us the use of Eq.~\eqref{eqn: H_star} in other contexts \cite{zurek2001information, Horodecki_2003}.
This work is supported by the National Research Foundation, Singapore and A*STAR under its CQT Bridging Grant, the Government of Spain (Severo Ochoa CEX2019-000910-S, Torres Quevedo PTQ2021-011870, TRANQI and European Union NextGenerationEU PRTR-C17.I1), Fundació Cellex, Fundació Mir-Puig, Generalitat de Catalunya (CERCA program), the European Union (QSNP, 101114043 and Quantera project Veriqtas),  the ERC AdG CERQUTE, the AXA Chair in Quantum Information Science, and the  European  Union’s  Horizon~2020  research and innovation programme under the Marie Sk\l{}odowska-Curie grant agreement No.~754510.


%

\appendix

\section{Entropy definitions}
\label{app:entropy_defs}
Since we will use them in more than one section of this Appendix, here we state the definitons of the von Neumann conditional entropy and the min- and max-entropies of bipartite states $\rho_{AE}$ \cite{renner2006security}.  
\begin{defn}\label{defn: H}
    The conditional entropy of $\rho_{AE}$ is defined by
\begin{align}
    H(A|E)_{\rho_{AE}} := S(\rho_{AE})- S(\rho_E)\,,
\end{align}
where $S(\rho)= -\tr \big( \rho \log_2 \rho \big)$ is the von Neumann entropy of $\rho$ and $\rho_{E}:= \tr_{A}{\rho_{AE}}$.
\end{defn}
\begin{defn}\label{defn: Hmin}
    The conditional min-entropy of $\rho_{AE}$ is defined by
    \begin{align}
    H_{\min}(A|E)_{\rho_{AE}} := \max_{\sigma_{E
}}\,\, H_{\min}(\rho_{AE}|\sigma_{E})\, 
\end{align}
with 
\begin{align}
H_{\min}(\rho_{AE}|\sigma_{E}):=- \min \{\lambda~|~2^\lambda (\id_{A}\otimes \sigma_{E}) \geq \rho_{AE}\}\,.
\end{align}
\end{defn}

\begin{defn}\label{defn: Hmax}
The conditional max-entropy of $\rho_{AE}$ is
\begin{align}\label{eq:relative-min-entropy}
    H_{\max}(A|E)_{\rho_{AE}}:= \max_{\sigma_{E
}}\,\,\log_{2} \tr \Big( (\id_{A}\otimes\sigma_{E}) \Pi_{AE} \Big)\,,
\end{align}
where $\Pi_{AE}$ is the projector onto the support of $\rho_{AE}$.     
\end{defn}
Notice that when the system $E$ is trivial (i.e. its Hilbert space is one-dimensional), 
$H_{\max}(A|E)_{\rho_{AE}}=H_{\max}(A)_{\rho_A}= \log_{2}\textrm{rank}(\rho_{A})$. In the following, when the state $\rho_{AE}$ to which we refer is clear from the context, we will drop the corresponding subscript in the notation for the conditional entropies.

\medskip

\par In our state discrimination scenario where Alice and Eve share a bipartite state $\rho_{AE}$, given any POVM $\mathcal{M}=\{M_x\}_x$ for Alice, one can define a classical-quantum state (cq-state) $\rho_{XE}$ to model the correlations between the measurement outcomes (classical information) and the corresponding post-measurement states $\rho_{E}^x$ on Eve's subsystem,
\begin{align}\label{eqn: cq-state}
 \rho_{XE}= \sum_x p_x\ketbra{x}{x}\otimes \rho_{E}^{x}\,,  
\end{align}
where $p_x=\tr[M_x\rho_A]$, $\rho_E^x=\tr_{A}[(M_x\otimes\id_E)\rho_{AE}]/p_x$ and $\{\ket{x}\}_x$ is some orthonormal basis representing Alice's outcomes. When the POVM $\mathcal{M}$ is extremal,
we have the following relation between $H_{\textrm{min}}(X|E)$ and Eve's optimal guessing probability given $\rho$ and $\mathcal{M}$ \cite{Konig_2009,Senno_2023}: 
\begin{align}\label{eqn: H and P}
H_{\textrm{min}}(X|E) = -\textrm{log}_{2}P_{\textrm{guess}}(\rho, \mathcal{M})\,. 
\end{align}

\section{Proof of optimality of rank-one measurement operators}
In this appendix we prove the optimality of rank-one measurements for $H^{*}_{\textrm{min}}$ and $H^{*}$. Similar results have been obtained for other information-theoretic quantities (see, e.g., \cite[Section II.I]{Modi_2012} in the context of quantum discord).

\subsection{Optimality of rank-one measurements for $H^{*}_{\textrm{min}}$ and $H^{*}_{\textnormal{max}}$}
\label{app:rank1}
We show that Alice's optimal measurement can be assumed to be rank-one. First, notice that any measurement can be obtained by coarse-graining a rank-one measurement, since, given some coarse-grained (i.e. not rank-one) measurement $\mathcal{M}_\textrm{coarse}=\{M_i\}_i$, we can represent each of its elements in its spectral decomposition as $M_i= \sum_{j}\lambda_{ij}\ketbra{f^i_j}$, with $\lambda_{ij} \geq 0$ for all $i$, $j$. $\mathcal{M}_\textrm{coarse}$ can then be seen as a coarse-graining of the fine-grained measurement $\mathcal{M}_{\textrm{fine}}=\{\ketbra{f^i_j}\}_{i,j}$. 
If we restrict Alice to performing some projective measurement $\mathcal{M}_{\textrm{coarse}}$, 
the corresponding $\mathcal{M}_{\textrm{fine}}$ will also be projective.
\par The coarse-graining of $\mathcal{M}_\textrm{coarse}$ can be seen as a deterministic post-processing of the outcomes of $\mathcal{M}_{\textrm{fine}}$. Consider the classical-quantum state \eqref{eqn: cq-state} formed by the classical information of Alice's measurement outcomes and Eve's quantum states in the case of $\mathcal{M}_{\textrm{fine}}$. The cq-state for any coarse-graining of $\mathcal{M}_{\textrm{fine}}$ can be found by applying a deterministic function $f(x)$ to the classical register $X$,
  \begin{align}\label{eqn: coarse cq-state}
 \rho_{XE}^{\textrm{coarse}}= \sum_{x}p_x\ketbra{f(x)}\otimes \rho_{E}^{x}\,.  
\end{align}  
Applying a function to a classical register cannot increase either the conditional min-entropy or the conditional max-entropy (see, e.g., \cite[Proposition 6.20]{Tomamichel_2016}). Therefore,
\begin{align}
H_{\textrm{min}}(X|E)_{\rho_{XE}^{\textrm{coarse}}} &\leq H_{\textrm{min}}(X|E)_ {\rho_{XE}^{\textrm{fine}}}\,,
\\
H_{\textrm{max}}(X|E)_{\rho_{XE}^{\textrm{coarse}}} &\leq H_{\textrm{max}}(X|E)_ {\rho_{XE}^{\textrm{fine}}}\,.
\end{align}
Then, by \eqref{eqn: H and P}, it is optimal for Alice to choose a rank-one measurement in order to minimise Eve's guessing probability and to maximise the conditional max-entropy of her measurement outcomes.

\subsection{Optimality of rank-one measurements for $H^{*}$}
\label{app:rank1_H}
We know from \cite[Theorem 1, (i)]{Coles_2012} that the conditional entropy of the classical-quantum state formed by Alice's outcomes $Z$ and Eve's post-measurement states is 
\begin{align}
    H(Z|E)= D\big(\rho\,||\sum_{i}M_i \rho M_i\big)\,,
\end{align}
where $\{M_i\}_i$ is Alice's projective measurement and $D(\rho||\sigma)$ is the quantum relative entropy between the states $\rho$ and $\sigma$,
\begin{align}
    D(\rho||\sigma)= \tr \Big( \rho \big(\log_2{\rho}-\log_2{\sigma} \big)\Big)\,,
\end{align}
and is defined when the support of $\rho$ is contained within the support of $\sigma$. $D\big(\rho\,||\sum_{i}M_i \rho M_i\big)$ can be written as 
\begin{align}\label{eqn: D}
   D\big(\rho\,||\sum_{i}M_i \rho M_i\big)= -S(\rho) + \tr \Big(\rho \log_2 \big( \sum_{i}M_i \rho M_i\big)\Big)\,. 
\end{align}
Since $\{M_i\}_i$ is projective, we have $\log_2 \big(\sum_{i}M_i \rho M_i \big) = \sum_i \log_2 \big(M_i \rho M_i\big)$ and $\Big(\log_2\big({M_i \rho M_i}\big)\Big)M_j= \delta_{ij}\log_2\big({M_i \rho M_i}\big)$, so the second term on the r.h.s. of \eqref{eqn: D} is
\small
\begin{align}
 \tr \Big(\rho \sum_{i} \log_2 \big(M_i \rho M_i\big)\Big) &=  \tr \Big(\sum_j M_j \rho M_j \sum_{i} \log_2 \big(M_i \rho M_i \big)\Big) \nonumber
 \\
 &=  -S\big(\sum_{i}M_i \rho M_i\big)\,,
\end{align}
\normalsize
and we recover Equation \eqref{eqn: S_ents}. To show that it is optimal for Alice to perform a rank-one measurement $\{M^{\textrm{fine}}_{ij}=\ketbra{f^i_j}\}_{i,j}$ rather than a coarse-grained one $\{M^{\textrm{coarse}}_i= \sum_{j}\lambda_{ij}\ketbra{f^i_j}\}_i$, denoting $\rho_{\textrm{fine}}=\sum_{i, j}M_{ij}\rho M_{ij}$ and $\rho_{\textrm{coarse}}=\sum_{i}M^{\textrm{coarse}}_i \rho M^{\textrm{coarse}}_i$, it is sufficient to show that 
\begin{align}\label{eqn: S_ineq}
    S(\rho_{\textrm{fine}}) \geq S(\rho_{\textrm{coarse}})\,.
\end{align}
Note that the state $\rho^{\textrm{fine}}$ is the average state after performing the measurement $\{M^{\textrm{fine}}_{ij}\}_{i,j}$ on $\rho_{\textrm{coarse}}$.
Projective measurements cannot increase the von Neumann of a state (see, e.g., \cite[Theorem 11.9]{nielsen00}), so the inequality \eqref{eqn: S_ineq} holds and it is optimal for Alice to perform a rank-one measurement to maximise the conditional entropy.

\section{Technical steps in the proof of Theorem \ref{thm: main}}
\label{app:proof1}

\subsection{Proof of Lemma \ref{lemma: measure} sketched in the main text}
\label{app:proof11}

The main steps of the proof of Lemma \ref{lemma: measure} sketched in the text are given here in the form of two lemmas.   

\begin{lemma}\label{lemma1}
Since $\mathcal{M}$ is rank-one, the $d$ states $\rho_i$ can be taken pure, $\rho_i=q_i\dyad{\psi_i} := |\tilde{\psi}_i\rangle\langle \tilde{\psi}_i|$. Hence,
\ba\label{eqn: PG0app}
P_{\textnormal{guess}}(\rho,\mathcal{M})=\max_{\{\ket{\tilde{\psi}_i}\}}\sum_{i}|\langle{m_i}|{\tilde{\psi_i}}\rangle|^2\,,
\ea 
with the constraint that $\rho = \sum_{i=1}^ d |\tilde{\psi}_i\rangle\langle \tilde{\psi}_i|$.
\end{lemma}

\begin{proof}
We consider that Alice and Eve share a pure state $\ket{\Phi}_{AE}$ such that $\textrm{Tr}_E\dyad{\Phi}=\rho_A :=\rho$, and Eve performs a measurement $\mathcal{N}=\{N_i\}_i$. If assume that Eve performs her measurement before Alice, we find that she steers Alice's state $\rho_i$ with probability $p_i=\matrixel{\Phi}{\mathbb{1}\otimes N_i}{\Phi}$ and we recover \eqref{eqn: sdp}. But, because this is a no-signalling scenario, the order of their measurements does not matter, so we could equally think of Alice as measuring first. When Alice measures outcome $i$, since every $M_i$ is rank-one, she steers the pure state $\ket{\phi_i}\propto \tr_A(M_i\otimes\mathbb{1}\ket{\Phi})$. Thus
\ba\label{eqn: sdp2}
P_{\textrm{guess}}(\rho,\mathcal{M})&=&\max_{\mathcal{N}}\sum_i\textrm{Tr}(\rho M_i)\matrixel{\phi_i}{N_i}{\phi_i} \\
&\textrm{s.t.}& N_i\geq 0\,,\;\sum_i N_i=\mathbb{1}\,. \nonumber
\ea 
The SDP \eqref{eqn: sdp2} describes the optimal discrimination by Eve of an ensemble compatible with her reduced state. It was shown in \cite{eldar2003designing} that the optimal measurement to distinguish an ensemble of pure states is made of $d$ rank-one operators. Thus, Eve will also be performing a rank-one measurement on her system, and the states $\rho_i$ steered on Alice's side will be pure.\end{proof}

\begin{lemma}\label{lemma2}
We will show that one can always write $\tilde{\ket{\psi}_i}=\sqrt{\rho}\ket{i}$, with $\{\ket{i}\}$ an orthonormal basis [Eq.~\eqref{decompEve} of the main text], and that Eq.~\eqref{PG1} follows as a consequence, namely
\ba\label{PG1app}
    P_{\textnormal{guess}}(\rho,\mathcal{M})&=&\max_{\{\ket{i}\}}  \sum_i \abs{\matrixel{m_i}{\sqrt{\rho}}{i}}^2\,.
\ea 
\end{lemma}

\begin{proof}

Following \cite{hughston1993complete}, given two density matrices $\rho'$ and $\rho$, there is a one-to-one map
\ba
\tilde{\ket{\psi_i'}}=\rho'^{1/2}\rho^{-1/2}\tilde{\ket{\psi_i}}\,, 
 \ i=1,...,n
\ea
between any two sub-normalised decompositions $\{\tilde{\ket{\psi_i'}}\}$ and $\{\tilde{\ket{\psi_i}}\}$ into the same number $n$ of pure states. Choosing $\rho'=\mathbb{1}/d$, we obtain $\tilde{\ket{\psi_i}}=\rho^{1/2}\tilde{\ket{i}}$, where $\sum_i \dyad{\tilde{i}}=\mathbb{1}$. Since we start with a decomposition of $\rho$ into $n=d$ pure states, and since any decomposition of $\mathbb{1}$ into $d$ pure states defines an orthonormal basis,
we find that
\begin{equation}\label{decomp}
    \tilde{\ket{\psi_i}}=\sqrt{\rho}\ket{i}
\end{equation} for some $\{\ket{i}\}$ forming an orthonormal basis. Plugging this into \eqref{eqn: PG0app}, we prove \eqref{PG1app}.
\end{proof}

\subsection{Optimisation \eqref{PG5} in the proof of Lemma \ref{lemma: measure}}
\label{app:thm2}
The optimisation \eqref{PG5} is a special case of the following: given two orthonormal bases $\{\ket{i}\}_{i=1,...,d}$ and $\{\ket{m_i}\}_{i=1,...,d}$ and an operator $A$ ($\sqrt{\rho}$ in our case) satisfying $A=A^\dagger\geq 0$ such that
\ba\label{constraintA}
\matrixel{m_i}{A}{m_i}=\frac{1}{d}\tr(A)\;\textrm{for all } i=1,...,d\,,
\ea
we want to compute
\ba\label{optim}
P&=&\max_{\{\ket{i}\}}\sum_i |\matrixel{i}{A}{m_i}|^2\,.
\ea 
This expression is equivalent to the probability of correct state discrimination of an ensemble $\sigma_i=\dyad{\tilde{\gamma_i}}$, with $\ket{\tilde{\gamma}_i}=A\ket{m_i}$, using the measurement $\{\Pi_i\}_{i=1,...,d}$.
A measurement $\{\Pi_i\}_{i=1,...,d}$ is optimal to distinguish a set of states $\{\sigma_i\}_{i=1,...d}$ if and only if $Y \equiv \sum_i \sigma_i\Pi_i\geq \sigma_j$ for all $j=1,...,d$ \cite{holevo1973statistical,helstrom1969quantum,yuen1975optimum}. In our case, $\{\ket{i}\}$ achieves the maximisation \eqref{optim} if and only if
\begin{align}\label{condthm2}
\sum_i A\dyad{m_i}A\dyad{i}-A\dyad{m_j}A\geq 0&\,\,\textrm{for all $j$}\,.
\end{align}
Let us make the \textit{guess} that, under the condition \eqref{constraintA}, the optimal measurement is given by $\ket{i}=\ket{m_i}$. With this, conditions \eqref{condthm2} read
\ba\label{condthm2_hold}
\frac{\tr(A)}{d}\,A\,-\,A\dyad{m_j}A := B_j\geq 0&\textrm{for all $j$}\,.
\ea It is clear that $B_j=B_j^\dagger$. Moreover, using \eqref{constraintA} we see that
\ban
\matrixel{\phi}{B_j}{\phi}&=&\matrixel{m_j}{A}{m_j}\matrixel{\phi}{A}{\phi}-\matrixel{\phi}{A}{m_j}\matrixel{m_j}{A}{\phi}
\ean for any vector $\ket{\phi}$. The r.h.s.~is always non-negative when $A=A^\dagger\geq 0$, because of the Cauchy-Schartz inequality applied to the vectors $\sqrt{A}\ket{\phi}$ and $\sqrt{A}\ket{m_j}$. This proves that $B_j\geq 0$ for all $j$, which is the desired condition, vindicating our guess. 

\subsection{Necessary condition for $P^{*}_{\textnormal{guess}}$}\label{app: nec_cond}
Here we show that satisfying condition~\eqref{eqn: equality} is  
necessary 
for Alice's measurement basis $\{\ket{m_i}\}$ to achieve the optimal guessing probability.
From \eqref{PG1}, choosing the orthonormal basis $\{\ket{i}\}=\{\ket{m_i}\}$, we have the bound 
\begin{align}\label{eqn: ep_bound}
P_{\textrm{guess}}(\rho,\mathcal{M}) \geq  \sum_i \abs{\matrixel{m_i}{\sqrt{\rho}}{m_i}}^2 \,.
\end{align}
Define the set of 
real numbers 
$\varepsilon_i:= \frac{\tr \sqrt{\rho}}{d} - \matrixel{m_i}{\sqrt{\rho}}{m_i}$. 
\\
Note that the constraint 
\begin{align}
    \tr\sqrt{\rho}= \sum_{i} \matrixel{m_i}{\sqrt{\rho}}{m_i} = \tr\sqrt{\rho}- \sum_{i} \varepsilon_i
\end{align}
implies that $\sum_i \varepsilon_i=0$.
In terms of $\{\varepsilon_i\}$, the bound in \eqref{eqn: ep_bound} is
\begin{align}
P_{\textrm{guess}}(\rho,\mathcal{M}) &\geq  \sum_i \bigg({\frac{\tr\sqrt{\rho}}{d}-\varepsilon_i}\bigg)^2 \nonumber
\\
&= \frac{(\tr\sqrt{\rho})^2}{d}-\frac{2\tr\sqrt{\rho}}{d}\sum_i \varepsilon_i+\sum_i \varepsilon_i^2 \nonumber
\\
&= P^{*}_{\textrm{guess}}(\rho) + \sum_i \varepsilon_i^2\,,
\end{align}
so the optimal $P^{*}_{\textrm{guess}}(\rho)$ cannot be achieved if any $\varepsilon_i \neq 0$, i.e. if the condition
\ba
\matrixel{m_i}{\sqrt{\rho}}{m_i}=\frac{1}{d}\tr\sqrt{\rho}&\textrm{for all}& i=1,...,d\,
\ea
is not satisfied.

\subsection{Eve's optimal measurement for $H^{*}_{\textnormal{min}}$}\label{app: Eve_discrim}
Denote the pure state shared by Alice and Eve by $\ket{\Psi}_{AE}$. Its Schmidt decomposition is
\begin{align}
    \ket{\Psi}_{AE}= \sum_{k=1}^{d}\sqrt{\lambda_k}\ket{u_k}_A\ket{v_k}_E\,,
\end{align}
where $\rho=\sum_{k=1}^{d}\lambda_k\ketbra{u_k}{u_k}$ and $\{\ket{v_k}\}$ is an orthonormal basis 
in which Eve's reduced state in diagonal. We can assume without loss of generality that
Eve's subsystem has the same dimension as Alice's subsystem, as Eve does not gain any advantage in discriminating Alice's states by holding 
a system of a higher dimension. 
Let Eve's local eigenbasis $\{\ket{v_k}\}$  
be related to that of Alice by the unitary $\ket{v_k}=U\ket{u_i}$. Then we can write Eve's 
reduced state as $\sigma_{E}=U\rho\, U^{\dagger}$. After Alice performs the rank-one measurement in the basis $\{\ket{m_i}\}$, Eve receives the subnormalised pure states $\ket{\tilde{\gamma}_i}= \sum_{k=1}^{d} \sqrt{\lambda_k}\langle{m_i|u_k} \rangle\ket{v_k}$, which can also be represented as 
\begin{align}
    \ket{\tilde{\gamma}_i}= \sum_{k=1}^{d} \sqrt{\lambda_k}\langle{m_i|u_k} \rangle\, U\ket{u_k}= U \sqrt{\rho}\ket{m_i^{*}}\,,
\end{align}
where $\ket{m_i^{*}}$ is the complex conjugate of $\ket{m_i}$ in the eigenbasis of $\rho$. Eve should then use an optimal rank-one measurement to discriminate
between the possible states, $\ket{\tilde{\gamma}_i}$, of her system. 
Note that, since Alice chooses an optimal measurement,
\begin{align}
  \frac{\tr\sqrt{\rho}}{d}&=\langle{m_{i}|\sqrt{\rho}|m_i} \rangle=\langle{m_{i}^{*}| \sqrt{\rho} |m_i^{*}}\rangle\,. \nonumber
\end{align}
In Appendix \ref{app:thm2}, we showed that a measurement in an orthogonal basis $\{\ket{n_i}\}$ is optimal to discriminate an ensemble of states $\{\ket{\tilde{\gamma}_i}=A\ket{n_i}\}$ if the condition $\langle n_i|A |n_i \rangle=\frac{1}{d}\tr A$ holds for all $i$, where $A \geq 0$. Here Eve receives the states $\ket{\tilde{\gamma}_i}=U\sqrt{\rho}\ket{m_i^{*}}=A\ket{n_i}$, where $A:= U \sqrt{\rho}\,U^{\dagger} \geq 0$ and $\ket{n_i}:=U\ket{m_i^{*}}$. We have $\frac{1}{d}\tr A= \frac{1}{d}\tr \big(U \sqrt{\rho}\, U^{\dagger}\big)=\frac{1}{d}\tr \sqrt{\rho}$ and 
\begin{align}
  \langle n_i | A | n_i\rangle= \langle m_i^{*}|\sqrt{\rho}|m_i^{*}\rangle=\frac{1}{d}\tr \sqrt{\rho}=\frac{1}{d}\tr A\,,  
\end{align}
so it is optimal for Eve to measure in the basis $\ket{n_i}=U\ket{m_i^{*}}$. In the case where $\rho$ is full-rank, we have 
\begin{align}
U\ket{m_i^{*}}=U \rho^{-1/2} U^{\dagger} \ket{\tilde{\gamma}_i}=\sigma^{-1/2}_{E}\ket{\tilde{\gamma}_i}\,,    
\end{align}
so Eve is performing a `pretty good' measurement.

\section{Alternative proof of Theorem \ref{thm: main}}
\label{app:proof2}
\subsection{Proof of the bound \eqref{mainlemma1} on $\pg$}\label{app: bound}
Here we provide an alternative derivation of the lower bound \eqref{mainlemma1} for Eve's optimal guessing probability. We make use of the min- and max-entropies defined in Appendix \ref{app:entropy_defs} and of the following lemma, which is a slight variation \footnote{In \cite[Lemma 3.1.13]{renner2006security}, on which Lemma \ref{lemma:renner} is based, Eq. \eqref{eq:min-entropy-ineq} is stated in terms of $H_{\textrm{min}}(\rho^{AE}|\sigma^{E})$ and $\geq H_{\textrm{min}}(\tilde{\rho}^{AE}|\sigma^{E})$, and for any $\sigma_E$ (c.f. Eq. \eqref{eq:relative-min-entropy}). It is easy to see that \eqref{eq:min-entropy-ineq} also holds by letting $\sigma_E$ be one state achieving the maximum in $H_{\textrm{min}}(A|E)_{\tilde{\rho}^{AE}}$.} of \cite[Lemma 3.1.13]{renner2006security}:
\begin{lemma}\label{lemma:renner}
Let $\{\ket{x}\}_x$ be an orthonormal basis on $\mathcal{H}_X$, let $\{\ket{\psi^x}\}_x$ be a family of unnormalized
vectors on $\mathcal{H}_A \otimes \mathcal{H}_E$, and define
\begin{align*}
\rho_{AE}&:=\ketbra{\psi}\textnormal{ with }\ket{\psi}=\sum_x\ket{\psi_x}\,,\\
\tilde{\rho}_{XAE}&:=\sum_x \ketbra{x}\otimes \ketbra{\psi_x}\,.
\end{align*}
Then, 
\begin{align}\label{eq:min-entropy-ineq}
H_{\textnormal{min}}(A|E)_{\rho_{AE}} \geq H_{\textnormal{min}}(A|E)_{\tilde{\rho}_{AE}}-H_{\textnormal{max}}(X)\,.
\end{align}
\end{lemma}
Consider a pure bipartite state $\rho_{AE}=\ket{\psi}\bra{\psi}$. Given any orthogonal basis $\{\ket{u_x}\}_x$ for $\mathcal{H}_{A}$, some orthogonal basis $\{\ket{v_y}\}_y$ of $\mathcal{H}_{E}$ can always be found such that the state vector $\ket{\psi}$ can be written as 
\begin{align}\label{eqn: ineq_Hmin}
    \ket{\psi}= \sum_{x, y}\alpha_{xy}\ket{u_x}\ket{v_y}= \sum_{x}\ket{u_x}\sum_{y}\alpha_{xy}\ket{v_j} := \sum_{x}\ket{\psi^{x}}\,.
\end{align} 
If the rank-one projective measurement $\{\ketbra{u_x}\}_x$ is performed on $A$ and the outcome $x$ is obtained, the (unnormalized) post-measurement state of $AE$ is
\begin{align*}
    \Big(\ketbra{u_x}{u_x} \otimes \id \Big) \rho^{AE} \Big(\ketbra{u_x}{u_x} \otimes \id \Big) &=  \ketbra{\psi^x}{\psi^x}.
\end{align*}
Therefore, the cq-state representing the correlations between the outcomes and the post-measurement states on $\mathcal{H}_A\otimes\mathcal{H}_E$ is, precisely as in Lemma \ref{lemma:renner},
\begin{align*}
\tilde{\rho}_{AEX} = \sum_x \ketbra{\psi^x}\otimes \ketbra{x}\,..
\end{align*}

To prove our target bound \eqref{mainlemma1}, we let $\rho_{AE}$ be the initial pure state shared by Alice and Eve and $\mathcal{M}=\{\ketbra{u_x}\}_x$ be Alice's rank-one projective measurement. Noting that $H_{\textrm{max}}(X)= \log_2\textrm{rank}(\tilde{\rho}_X) \leq \log_2d$
and that for any pure state $\rho_{AE}$
\begin{align}
  H_{\textrm{min}}(A|E)= -2 \log_{2} \sqrt{\rho_{A}}\,,   
\end{align}
returning to our usual notation with $\rho_{A}= \rho$ and $d_{A}=d$, we get from Lemma \ref{lemma:renner} the following upper bound on $H_{\textrm{min}}(A|E)_{\tilde{\rho}_{AE}}$, 
\begin{align}\label{eqn: final_ineq_Hmin}
\log_{2}d-2 \log_{2} \tr \sqrt{\rho} \geq H_{\textrm{min}}(A|E)_{\tilde{\rho}_{AE}}\,.   \end{align}

Finally, notice that, since Alice's post-measurement states form an orthonormal basis, the states $\tilde{\rho}_{XE}$ and $\tilde{\rho}_{AE}$ are related by a unitary on the first subsystem. Therefore, using twice the data processing inequality for conditional min-entropies \cite{renner2006security}, one can replace $H_{\textrm{min}}(A|E)_{\tilde{\rho}_{AE}}$ with $H_{\textrm{min}}(X|E)_{\tilde{\rho}_{XE}}$ in \eqref{eqn: final_ineq_Hmin} and, together with \eqref{eqn: H and P}, obtain 
\begin{align}\label{eqn: p_bound}
P_{\textrm{guess}}(\rho, \mathcal{M}) \geq \frac{1}{d}\Big(\tr \sqrt{\rho} \Big)^2.  
\end{align}

\subsection{Proof that the bound \eqref{mainlemma1} can be reached}
Here, we provide an alternative proof using semidefinite programming that the lower bound \eqref{mainlemma1} for Eve's optimal guessing probability can be reached when Alice performs measurements satisfying \eqref{eqn: equality}.
We consider the maximisation problem \eqref{eqn: sdp} for the guessing probability and reduce our study to the set of projective rank-one measurements $M_i=\dyad{m_i}$ in dimension $d$ that satisfy $\tr\Big(\ketbra{m_i}{m_i} \sqrt{\rho} \Big)= \frac{1}{d}\tr\sqrt{\rho}$ for all $\ket{m_i}$ (the set is certainly not empty, since all measurements in a basis that is unbiased with the eigenbasis of $\rho$ satisfy this condition).
From hereon, we assume without loss of generality that $\rho$ is full-rank, as we note that only the projection of the measurement elements $M_i$ onto the support of $\rho$ plays a role in \eqref{eqn: sdp}. More precisely, defining $\Pi_{\rho}$ as the orthogonal projector onto the support of $\rho$ and $M_{i}':= \Pi_{\rho}M_i \Pi_{\rho}$, we have 
\begin{align}
\max_{\{\rho_i\}}\sum_i \tr(M_i\rho_i)&= \max_{\{\rho_i\}}\sum_i \tr(M_i\,\Pi_{\rho}\rho_i\Pi_{\rho}) \nonumber
\\
&= \max_{\{\rho_i\}}\sum_i \tr(M_i'\rho_i)\,.    
\end{align}
Note too that 
\begin{align}
\tr\Big(M_i \sqrt{\rho} \Big)= \tr\Big(M_i' \sqrt{\rho} \Big)=\frac{1}{d}\tr\sqrt{\rho}\,.    
\end{align}
Then, in the case where $\rho$ is not rank-one, we can consider the problem projected onto the support of $\rho$, using the rank-one (not necessarily projective) measurement with elements $M_i'=\dyad{m_i'}$ in \eqref{eqn: sdp}, where $M_i'$ is in dimension $r=\textrm{rank}(\rho)$, has $d$ outcomes and satisfies $\tr\Big(M_i' \sqrt{\rho} \Big)=\frac{1}{d}\tr\sqrt{\rho}$.
\par Now with the full-rank $\rho$ assumption, we note that, since \eqref{eqn: sdp} is a semidefinite programming problem, we can define its corresponding minimisation (or dual) problem, which is given by  
\begin{align}\label{eqn: min}
    \beta(\rho, \mathcal{M}) &= \underset{X}{\textrm{min}} \tr\big( X\rho\big)\,, \qquad \textrm{s.t.} \,\,\, X \geq M_i\,.
\end{align}
The set of states $\{\rho_{i}=\rho/d\}$ and the matrix $X=2\id$ define strongly feasible points (i.e. points that satisfy the necessary constraints with strict inequalities) on the dual and primal problems respectively, so \eqref{eqn: sdp} and \eqref{eqn: min} both return $P_{\textrm{guess}}(\rho, \, \mathcal{M})$. For further details on semidefinite programming problems, see, for example, \cite{Skrzypczyk_2023}. Given the measurement $\mathcal{M}$, we can use the primal and dual problems to set upper and lower bounds respectively on Eve's optimal guessing probability, 
\begin{align}\label{eqn: bounded_P}
    \sum_{i}\tr\Big(\rho_{i} M_i \Big) \leq \pg(\rho,\mathcal{M}) \leq \tr\Big( X \rho \Big)\,,
\end{align}
where $\{\rho_i\}$ is any set of subnormalised states satisfying $\sum_{i}\rho_{i}=\rho$ and $X$ is any positive-semidefinite matrix satisfying $X - M_i \geq 0$ for all $M_i$. The set of states $\rho_i= \sqrt{\rho}\dyad{m_i}\sqrt{\rho}$ recovers the lower bound \eqref{mainlemma1}. 
We now show that the matrix $X=\frac{\tr\sqrt{\rho}}{d}\rho^{-\frac{1}{2}}$ 
satisfies 
\begin{align}\label{eqn: Xmat}
    X - \dyad{m_i} = \frac{\tr\sqrt{\rho}}{d}\rho^{-\frac{1}{2}}- \dyad{m_i} \geq 0 \qquad \forall\,\, i\,.
\end{align}
We can define any vector in dimension $d$ as $\sqrt{\rho}\ket{\phi}$ for some $\ket{\phi}$. To prove \eqref{eqn: Xmat}, it suffices to show that
\begin{align}\label{eqn: Xpos}
 &\matrixel{m_i}{\sqrt{\rho}}{m_i}\matrixel{\phi}{\sqrt{\rho}}{\phi}-\matrixel{\phi}{\sqrt{\rho}}{m_i}\matrixel{m_i}{\sqrt{\rho}}{\phi} \geq 0\,.  
\end{align}
In analogy with the proof in \ref{app:thm2}, we see that \eqref{eqn: Xpos} is true by applying the Cauchy--Schwarz inequality on the vectors $\rho^{\frac{1}{4}}\ket{m_i}$ and $\rho^{\frac{1}{4}}\ket{\phi}$. 

Using $X=\frac{\tr\sqrt{\rho}}{d}\rho^{-\frac{1}{2}}$ in \eqref{eqn: bounded_P}, we find 
\begin{align}
 \frac{1}{d}\Big(\tr \sqrt{\rho} \Big)^2 \leq P_{\textnormal{guess}}(\rho,\mathcal{M}) \leq \frac{1}{d}\Big(\tr \sqrt{\rho} \Big)^2\,.   
\end{align}
This shows that the bound \eqref{mainlemma1} can be saturated by measurements satisfying \eqref{eqn: equality}, so we have that the optimal guessing probability for Alice given the state $\rho$ is 
\begin{align}
  P_{\textrm{guess}}^{*}(\rho) =\frac{1}{d}\Big(\tr \sqrt{\rho} \Big)^2\,.     
\end{align}

\section{Technical details in the proof of Theorem \ref{thm: Hmax}}\label{app: Hmax_details}

We denote the Schmidt decomposition of the pure state $\ket{\Psi}_{AE}$ shared by Alice and Eve by
\begin{align}
    \ket{\Psi}_{AE}= \sum_{k=1}^{d}\sqrt{\lambda_k}\ket{u_k}_A\ket{v_k}_E\,,
\end{align}
where $\rho=\sum_{k=1}^{d}\lambda_k\ketbra{u_k}{u_k}$ and $\{\ket{v_k}\}$ is an orthonormal basis 
in which Eve's reduced state in diagonal. When Alice measures in an orthonormal basis $\{\ket{m_x}\}$ with dimension $d$, we can write Eve's post-measurement states $\{\ket{\psi_{x}^{E}}\}$ as
\begin{align}
\ket{\psi_{x}^{E}} = \frac{1}{\sqrt{p_x}} \langle m_x| \Psi_{AE} \rangle\,,  \quad p_x = \langle m_x | \rho | m_x \rangle  \,. 
\end{align}
Denoting Eve's average state post-measurement as $\tilde{\rho}_{E}$, i.e.
\begin{align}
  \tilde{\rho}_{E}=\sum_{x}p_x \ketbra{\psi_{x}^{E}}{\psi_{x}^{E}}\,,  
\end{align}
we note that 
\begin{align}
\tilde{\rho}_{E}= \sum_{x} \langle m_x | \Psi_{AE} \rangle \langle \Psi_{AE} | m_x\rangle = \tr_{A}\ketbra{\Psi_{AE}}{\Psi_{AE}} := \rho_E\,,
\end{align}
where $\rho_{E}$ is Eve's local state before the measurement. The following expression for $\hmax \left(X|E \right)$ in the special case of classical-quantum states is given in \cite{Konig_2009},
\begin{align}
    \hmax\left (A|E \right)= \log_2 \psecr\,,
\end{align}
where
\begin{align}\label{eqn: psec_Konig_app}
 \psecr = \max_{\sigma} \left( \sum_{x} \sqrt{p_x} F \left(\rho_{x}^{E}, \sigma \right) \right)^2\,,   
\end{align}
where $\sigma$ is a quantum state ($\sigma \geq 0$ and $\tr \sigma=1$) and
\begin{align}
    F \left(\rho, \sigma \right)= \tr \sqrt{\sqrt{\sigma} \rho \sqrt{\sigma} }
\end{align}
(note that $F \left(\rho, \sigma \right)$ is symmetric in $\rho$ and $\sigma$). Since, in our case, the states $\rho_{x}^{E}$ are pure, the expression \eqref{eqn: psec_Konig_app} reduces to 
\begin{align}\label{eqn: psecr_trace_app}
\psecr &= \max_{\sigma} \left( \sum_{x}  \sqrt{{p_x} \langle \psi_{x}^{E}|\sigma | \psi_{x}^{E}} \rangle
  \right)^2
\\
&= \max_{\sigma} \left( \sum_{x}  \sqrt{{p_x}\tr\left( { \sigma \ketbra{\psi_{x}^{E}}{\psi_{x}^{E}}}\right) } \right)^2\,.
\end{align}
We define by $\sigma^{*}$ the state $\sigma$ that achieves the maximisation in \eqref{eqn: psec_Konig_app}. By the Cauchy-Schwartz inequality, we have 
\begin{align}
  \psecr &= \left( \sum_{x}  \sqrt{{p_x}\tr \left( { \sigma^{*} \ketbra{\psi_{x}^{E}}{\psi_{x}^{E}}} \right)} \right)^2 
  \\
  &\leq d \sum_{x} p_x \tr \left( { \sigma^{*} \ketbra{\psi_{x}^{E}}{\psi_{x}^{E}}} \right) = d  \tr \left( \sigma^{*} {\rho}_E \right)\,. 
\end{align}
The inequality is saturated if and only if all of the terms inside the sum are identical, 
\begin{align}\label{eqn: CS_conds_app}
    p_x \tr \left( { \sigma^{*} \ketbra{\psi_{x}^{E}}{\psi_{x}^{E}}} \right) = \frac{\tr \big( \sigma^{*} {\rho}_E \big)}{d} \;\; \textrm{for all} \;\; x=1,...,d\,.
\end{align}
We find a second inequality for $\psecr$ by noting that 
\begin{align}\label{eqn: bound_E_app}
 \tr \big( \sigma^{*} {\rho}_E \big) \leq \max_{\sigma} \tr \big( \sigma {\rho}_E \big) = \lmax(\rho_{E})\,,   
\end{align}
where the second term is a known SDP that returns the largest eigenvalue of $\rho_E$. In the case where the largest eigenvalue of ${\rho}_E$  is non-degenerate, the maximisation is achieved only by $\sigma=\ketbra{\vmax}{\vmax}$, where $\ket{\vmax}$ is the vector from Eve's local basis $\{\ket{v_i}\}$ corresponding to the largest eigenvalue. Restricting for now to the case where the largest eigenvalue is non-degenerate, we see that the bound  \eqref{eqn: bound_E_app}  is reached if and only if $\sigma^{*}=\ketbra{\vmax}{\vmax}$ and \eqref{eqn: CS_conds_app} is satisfied. We are free to combine these conditions such that the necessary and sufficient conditions for the measurement basis $\{\ket{m_x}\}$ to achieve the bound \eqref{eqn: bound_E_app} are, from \eqref{eqn: CS_conds_app},
\begin{align}
    \abs{ \langle m_x | \umax \rangle  }^{2} = \frac{1}{d} \; \; \textrm{for all} \;\; x=1,...,d\,,
\end{align}
where $\ket{\umax}$ is the vector from Alice's local basis $\{\ket{u_i}\}$ corresponding to the largest eigenvalue. In the case where the largest eigenvalue of $\rho$ is degenerate, denote by $\{\ket{\vmax^{(i)}}\}$ the set of vectors in Eve's local basis corresponding to the largest eigenvalue. The optimal $\sigma$ in \eqref{eqn: bound_E_app} is now of the form 
\begin{align}
   \sigma=  \sum_{i} \gamma_i \ketbra{\vmax^{(i)}}{\vmax^{(i)}}\,, \quad \gamma_i \geq 0\,, \quad \sum_{i}\gamma_i=1\,.
\end{align}
Now the necessary and sufficient conditions for $\{\ket{m_x}\}$ to achieve the bound are
\begin{align}
    \sum_i {\gamma_i} \abs{\langle{m_x}|{\umax^{(i)}\rangle}}^{2}= \frac{1}{d} \;\; \textrm{for all} \;\; x=1,...,d\,,
\end{align}
where $\{\gamma_i\}$ is any set of non-negative numbers summing to 1 and $\{\ket{\umax^{(i)}}\}$ is the set of vectors in Alice's local basis corresponding to the maximum eigenvalue. 

Note, finally, that $\lmax(\rho_E)= \lmax(\rho)$, so in terms of Alice's state $\rho$ we have 
\begin{align}
    \psecr^{*} = d \lmax(\rho)
\end{align}
and 
\begin{align}
    \hmax^{*}\left (A|E \right) = \log_2 d + \log_2 \lmax(\rho)\,.
\end{align}

\section{Additional details from Section \ref{sec: case_studies}}\label{app: case_studies}

\subsection{Parameters for qutrit measurement \eqref{eqn: qutrit_meas}}\label{app: qutrit}
For convenience, we restate the measurement \eqref{eqn: qutrit_meas} in the following. Consider the measurement basis $\{M_i=\dyad{m_i}\}_{i=1,2,3}$, with 
\begin{align}
    \ket{m_1 }&=\sqrt{\frac{1+a}{3}}\ket{1}+\sqrt{\frac{1+b}{3}}\ket{2}+\sqrt{\frac{1+c}{3}}\ket{3}\,, \nonumber
    \\
    \ket{m_2 }&=\sqrt{\frac{1+a}{3}}\ e^{i\theta_1}\ket{1}+\sqrt{\frac{1+b}{3}}\ket{2}+\sqrt{\frac{1+c}{3}}\ e^{i\theta_2}\ket{3}\,,
\end{align}
and $\ket{m_3}$ defined by the normalization condition $\sum_i M_i=\id$, 
where $a = -({\gamma_2}-{\gamma_3})k,\, b = ({\gamma_1}-{\gamma_3})k,\, c = -({\gamma_1}-{\gamma_2})k,\,  k\in \mathbb{R}$ and each $\gamma_i \geq 0$ with $\gamma_1 \geq \gamma_2 \geq \gamma_3$. 
To ensure that the square root terms in the coefficients of $\ket{m_1}$ and $\ket{m_2}$ are well-defined, we need that $1+x \geq 0$, $x \in \{a, b, c\}$, which imposes the following constraint on $k$,
\begin{align}\label{eqn: k_first_cond}
  -\frac{1}{{\gamma_1}-{\gamma_3}}\leq k \leq
\frac 1{\max\{{\gamma_1}-{\gamma_2},{\gamma_2}-{\gamma_3}\}}\,.  
\end{align}
Furthermore, imposing $\langle{m_1|m_2}\rangle=0$ sets the following restriction, 
\begin{align}\label{eqn: abc_rest}
 \frac{1+a}{3}e^{i\theta_1}+ \frac{1+b}{3}+ \frac{1+c}{3}e^{i\theta_2} =0\,.   
\end{align}
Considering each of the terms in the sum as vectors in the complex plane, the sum $\frac{1+a}{3}e^{i\theta_1}+ \frac{1+b}{3}$ can take any absolute value in the range $\frac{1}{3}|b-a|$ to $\frac{1}{3}|2+a+b|$ by taking an appropriate choice of $\theta_1$. We then see that in order to satisfy \eqref{eqn: abc_rest}, $\frac{1+c}{3}$ must fall in this range, i.e. 
\begin{align}
    |b-a| \leq 1+c \leq 2 + a+ b\,.
\end{align}
This gives the following restriction on $k$,
\begin{align}
    - \frac{1}{2}\frac{1}{{\gamma_1}-{\gamma_{3}}} \leq k \leq \frac{1}{2}\frac{1}{{\gamma_1}-{\gamma_{3}}}\,,
\end{align}
which is a tighter bound than the constraint \eqref{eqn: k_first_cond}. Notice that the tightest constraint is $-\frac{1}{2}\leq k\leq \frac{1}{2}$ and is obtained for pure states. 

For $k \neq 0$, this measurement basis is unbiased to the eigenbasis of $\rho$ if and only if $a=b=c=0$ i.e. when $\rho$ is the maximally mixed state.
When we set $\{\gamma_i\}=\{\sqrt{\lambda_i}\}$, it is easy to show that $\matrixel{m_i}{\sqrt{\rho}}{m_i}=\frac{1}{3}\tr\sqrt{\rho}$ for $i=1,2,3$, so condition \eqref{eqn: equality} for maximal $H_{\textnormal{min}}$ is satisfied.
However, this measurement does not satisfy the necessary and sufficient condition \eqref{eqn: H equality} to maximise the conditional entropy, except in the special case where
\small
\begin{align}
\lambda_{1}(\sqrt{\lambda_{2}}-\sqrt{\lambda_3}) - \lambda_{2}(\sqrt{\lambda_1}+\sqrt{\lambda_{3}})-\lambda_{3}(\sqrt{\lambda_1}-\sqrt{\lambda_2})=0\,.  \end{align}
\normalsize
Moreover, we can also consider this qutrit measurement taking $\{\gamma_i\}=\{\lambda_i\}$. Similarly, it is straightforward to see that this measurement satisfies \eqref{eqn: H equality} and thus achieves the maximal conditional entropy $H^{*}$.
It does not, however, satisfy the necessary and sufficient condition \eqref{eqn: equality} to achieve $H^{*}_{\textnormal{min}}$ except in the special case where
\begin{align}
    \sqrt{\lambda_1}(\lambda_2-\lambda_3)-\sqrt{\lambda_2}(\lambda_1-\lambda_3)+\sqrt{\lambda_3}(\lambda_1-\lambda_2)=0\,.
\end{align}
Finally, when $\lambda_1 > \lambda_2$ (i.e. the largest eigenvalue of $\rho$ is not degenerate),  the condition \eqref{eqn: Hmax equality} for the measurement \eqref{eqn: qutrit_meas} to produce maximal $\hmax$ is satisfied if and only if $a=0$, so $\gamma_2=\gamma_3$. This condition is, in general, inequivalent to the conditions for maximal $H_{\textnormal{min}}$ and $H$.

\subsection{A two-qubit basis that has no unbiased product basis}\label{app:prodbasis}
We will show that there is no orthonormal basis of two-qubit product states that is unbiased to the basis:
\begin{equation}
    \begin{aligned}\label{eq:entbasis}
    \ket{\psi_1}&=\ket{00}\\
    \ket{\psi_2}&=\frac{1}{\sqrt{3}}(\ket{01}+\ket{10}+\ket{11})\\
    \ket{\psi_3}&=\frac{1}{\sqrt{3}}(\ket{01}+\omega\ket{10}+\omega^2\ket{11})\\
    \ket{\psi_4}&=\frac{1}{\sqrt{3}}(\ket{01}+\omega^2\ket{10}+\omega\ket{11}),
    \end{aligned}
\end{equation}
where $\omega=e^{2\pi i/3}$. Firstly, we note that any orthonormal two-qubit product basis can be expressed as either 
\begin{enumerate}[(i)]
    \item\label{case1} $\left\{\ket{a}\ket{A},\ket{a}\ket{A^\perp},\ket{a^\perp}\ket{B},\ket{a^\perp}\ket{B^\perp}\right\}$ or
    \item\label{case2} $\left\{\ket{a}\ket{A},\ket{a^\perp}\ket{A},\ket{b}\ket{A^\perp},\ket{b^\perp}\ket{A^\perp} \right\}$
\end{enumerate}
for some single qubit states $\ket{a},\ket{A},\ket{b},\ket{B}$, where $\perp$ denotes the unique orthogonal state to a given qubit state. Consider Case~\eqref{case1}. If this basis is unbiased to the vectors~\eqref{eq:entbasis}, we have 
\begin{equation}\label{eq:mu00}
    \abs{\braket{00}{aA}}^2=\abs{\braket{00}{aA^\perp}}^2=\frac14.
\end{equation}
Letting $\ket{A}=A_0\ket{0}+A_1\ket{1}$, etc, with $A_0\in\mathbb{R}$, $A_1\in\mathbb{C}$, it follows that $A_0=A^\perp_0={1}/{2a_0}$. In order for $\ket{A}$ and $\ket{A^\perp}$ to be orthogonal and have the same coefficient of $\ket{0}$, both states must lie on the $XY$-plane of the Bloch sphere, i.e. $\ket{A}=\frac{1}{\sqrt{2}}(\ket{0}+e^{iy}\ket{1})$ and $\ket{A^\perp}=\frac{1}{\sqrt{2}}(\ket{0}-e^{iy}\ket{1})$ for some $0\leq y<2\pi$. The relation between $a_0$ and $A_0$ implies that $\ket{a}$ also lies on this plane and we can take $\ket{a}=\frac{1}{\sqrt{2}}(\ket{0}+e^{ix}\ket{1})$ for some $0\leq x<2\pi$. Similar reasoning shows that all the states $\ket{a}, \ket{A}, \ket{b}, \ket{B}$ must lie in the $XY$-plane, in both Cases~\eqref{case1} and~\eqref{case2}.

We now show the states $\ket{aA}$ and $\ket{aA^\perp}$ cannot be mutually unbiased to both $\ket{\psi_2}$ and $\ket{\psi_3}$. Firstly, the equalities $\abs{\braket{\psi_2}{aA}}^2=\abs{\braket{\psi_2}{aA^\perp}}^2=1/4$ give 
\begin{align}\label{eqn: trig}
&\cos(x)+\cos(y)+\cos(x-y)=    \nonumber  \\ 
&\cos(x)-\cos(y)-\cos(x-y)=0\,.    
\end{align}
Thus, we find $\cos(x)=0$, giving four possibilities for the values of $(x,y)$: either $(\pi/2,3\pi/4)$, $(\pi/2,7\pi/4)$, $(3\pi/2,\pi/4)$, or $(3\pi/2,5\pi/4)$. Substitution shows that none of these pairs give $\abs{\braket{\psi_3}{aA}}^2=1/4$. A similar argument shows that bases of the form in Case~\eqref{case2} also cannot be mutually unbiased to the vectors~\eqref{eq:entbasis}.

\end{document}